\documentclass[11pt]{amsart}
\usepackage{a4wide,xcolor,eucal,enumerate,mathrsfs,graphicx,subfig,booktabs,verbatim}
\usepackage{pdfsync}
\usepackage[hidelinks]{hyperref}
\usepackage{amsmath,amssymb,epsfig,amsthm,bm}
\usepackage[latin1]{inputenc}
\usepackage{dsfont}

\usepackage[english]{babel}
\usepackage{amsmath,amsfonts,amsthm,amssymb,dsfont,url}
\usepackage[foot]{amsaddr}
\usepackage[babel]{csquotes}
\usepackage{caption}
\usepackage[numbers,sort&compress]{natbib}
\bibpunct[, ]{[}{]}{,}{n}{,}{,}
\makeatletter
\def\NAT@def@citea{\def\@citea{\NAT@separator}}
\makeatother

\theoremstyle{plain}
\newtheorem{theorem}{Theorem}[section]
\newtheorem{lemma}[theorem]{Lemma}

\theoremstyle{definition}

\theoremstyle{remark}

\numberwithin{equation}{section}
\makeatletter
\g@addto@macro{\endabstract}{\@setabstract}
\newcommand{\authorfootnotes}{\renewcommand\thefootnote{\@fnsymbol\c@footnote}}%
\makeatother

\begin{document}
\begin{center}
  \LARGE 
 Accurate approximation of the distributions of the 3D Poisson-Voronoi typical cell geometrical features \par \bigskip

  \normalsize
  \authorfootnotes
  Martina~Vittorietti\footnote{Contact author: m.vittorietti@tudelft.nl}\textsuperscript{1}\textsuperscript{2}, Geurt~Jongbloed\textsuperscript{1},
 Piet J.J.~Kok\textsuperscript{3}, Jilt~Sietsma\textsuperscript{4}\par \bigskip

  \textsuperscript{1}Department of Applied Mathematics, Delft University of Technology, Mekelweg 4, 2628 CD Delft, The Netherlands\par
  \textsuperscript{2}Materials Innovation Institute (M2i), Mekelweg 4, 2628 CD Delft, The Netherlands\par 
\textsuperscript{3}Tata Steel, IJmuiden Technology Centre, Postbus 10.000, 1970 CA, IJmuiden, The Netherlands\par 
\textsuperscript{4}Department of Materials Science and Engineering, Delft University of Technology, Mekelweg 2, 2628 CD Delft, The Netherlands\par \bigskip

\end{center}





\keywords{Voronoi; Poisson-Voronoi diagrams; 3D grain size; parametric approximation}

\date{}

\dedicatory{}

\begin{abstract}
Although Poisson-Voronoi diagrams have interesting mathematical properties, there is still much to discover about the geometrical properties of its grains. Through simulations, many authors were able to obtain numerical approximations of the moments of the distributions of more or less all geometrical characteristics of the grain. Furthermore, many proposals on how to get close parametric approximations to the real distributions were put forward by several authors.  In this paper we show that exploiting the scaling property of the underlying Poisson process, we are able to derive the distribution of the main geometrical features of the grain for every value of the intensity parameter.  Moreover, we use a sophisticated simulation program to construct a close Monte Carlo based approximation for the distributions of interest. Using this, we also determine the closest approximating distributions within the mentioned frequently used parametric classes of distributions and conclude that these approximations can be quite accurate.

\end{abstract}

\section{Introduction}

In the past few years the use of  Voronoi diagrams has rapidly increased. These diagrams represent an appealing structure, especially because they describe various natural processes quite well. In \cite{okabe2009spatial} an extensive list of fields in which Voronoi diagrams are adopted can be found.
Among the many areas of applications of this model, the field of materials science stands out.
In fact, they are now among the most used mathematical models for microstructure characterization and depending on the specific kind of materials, it is possible to use a proper category of Voronoi diagrams.\\
In this paper, we discuss the most basic instance of the model: Poisson-Voronoi diagrams.
In this framework the nuclei or sites are generated by a homogeneous Poisson process with intensity parameter $\lambda$.

Although many interesting mathematical properties of Poisson-Voronoi diagrams are known, there is still much to discover about the distributions of the geometrical characteristics of its grains.
Through simulations, many authors were able to obtain numerical approximations of the moments of the distribution of the volume, of the surface area, of the number of faces and many other geometrical characteristics of the grains. Nevertheless, analytic expressions of the distributions of many of these important features are not known, others are only known via complicated numerically intractable characterizations.
Therefore, various proposals to obtain close approximations to the real distributions were put forward by several authors e.g.\ Lognormal-, Generalized Gamma- and Rayleigh distributions. But as far as we know, there is no theoretical support for preferring one of these distributions.

The aim of this paper is twofold. After explaining that $\lambda$, the intensity parameter of the Poisson process, is the only parameter determining all distributional properties of the geometrical structure of the grain, we show that if we have the distribution of a given geometrical characteristic for $\lambda=1$, the distribution of the same quantity for every value of $\lambda>0$ can be obtained by rescaling. More precisely, we consider volume, surface area and number of faces of the grain, but the approach can be extended to other characteristics.
Secondly, we find a close Monte Carlo based approximation for the previously mentioned geometrical characteristics of the grains and using  it we determine the most closely approximating distribution within the mentioned frequently used parametric classes of distributions.
As said before, several well known probability distributions were used for approximating the grain geometrical characteristics distributions, but in this study we determine the `best' of these.

After briefly reviewing the basic concepts of Voronoi diagrams and the Poisson process in Section \ref{sec:PVdiagrams}, in Section \ref{sec:distrTyp} we explain the scaling property of the Voronoi structure in terms of the intensity parameter and how it can be useful for studying distributional properties of the grain features.
Since the intensity parameter $\lambda$ is the only parameter involved in generating a specific structure, it governs the distribution of all the geometrical characteristics of the Poisson-Voronoi typical cell. Later, we explain how the scaling acts on the different geometrical features and we show an empirical example of what happens changing the scale parameter.
Section \ref{sec:simulres}  describes our simulation approach and produces an accurate Monte Carlo approximation for the distribution of the grain volume and the grain surface area. In fact, we provide the approximate distributions of the volume and of the grain surface area for $\lambda=1$ and we can adapt it for the other values of $\lambda$ using the aforementioned scaling properties.
In Section \ref{sec:parappr}, we study how well the true distributions of the geometrical characteristics can be approximated by some well-known and frequently used probability distributions in this context: the Gamma-, Generalized Gamma- and Lognormal distribution.
Fitting these three distributions and comparing them through statistical measures such as the supremum distance between the Monte Carlo empirical distribution and its parametric approximations and Total Variation distance, we are not only able to identify the best approximation but also to give a measure of error if one of these parametric approximations is used.
Finally, we discuss the possibility to extend our approach according to different Voronoi Diagrams cases, such as Multi-level Voronoi and/or Laguerre Voronoi Diagrams. We want to remark that for the 3D Voronoi diagrams generation we use TATA Steel software and for data analysis {R}.

\section{Poisson-Voronoi Diagrams}
\label{sec:PVdiagrams}
We begin by reviewing the generic definition and the basic properties of the Poisson-Voronoi Diagram.
Given a denumerable set of distinct points in $\mathbb{R}^d$, $\textbf{\textrm{X}}=\{x_i: i\ge1\}$, the Voronoi diagram of $\mathbb{R}^d$ with \textit{nuclei} $\{x_i\}$ (also called \textit{sites} or \textit{generator points}) is a partition of $\mathbb{R}^d$ consisting of cells
  \begin{equation*}\notag
   C_i=\{ y\in\mathbb{R}^d\,:\, \| x_i-y \| \le \|x_j-y\|\,\, for\,\, j\ne i \},\,\,\, i=1,2,\dots
   \end{equation*}
   where $\|\cdot\|$ is the usual Euclidean distance.
   This means that given a set of two or more but finitely many distinct points, we associate all locations in that space with the closest member(s) of the point set with respect to the Euclidean distance.

If we assume that $\textbf{\textrm{X}}=\mathrm{\Phi}=\{x_i\}$ is the realization of a homogeneous Poisson point process, then we will refer to the resulting structure as the \textit{Poisson-Voronoi diagram}, $\mathcal{V}_\Phi$.

We find it useful to remind briefly what a Poisson process is and which are its basic properties. We follow Kingman's approach in the mathematical definition of the process \cite[cf.][]{kingman1993poisson}.

Let $S$ be a measurable set in $\mathbb{R}^d$, $\mathcal{B}_S=\mathcal{B}(S)$ the $\sigma$-field of its Borel sets in $S$ and $\mu$ a boundedly finite Borel measure on $\mathcal{B}_S$. Moreover, denote $N(A)=\texttt{\#}\{i:X_i\in A\}$.
A \textit{Poisson process} on $S$ is then a random countable subset $\Phi$ of $S$, such that
\begin{itemize}
\item for every finite family of disjoint bounded Borel subsets $A_1,A_2,\dots, A_n$ of $S$, the random variables $N(A_1),N(A_2),\dots, N(A_n)$ are independent
\item $N(A)$ has Poisson distribution $\mathcal{P}(\lambda)$, where $\lambda=\mu(A)$ lies in $0\le\lambda\le\infty$.
\end{itemize}
From this it immediately follows that
\begin{equation*}\notag
\mu(A)=\mathbb{E}\{N(A)\}.
\end{equation*}
Therefore the measure $\mu$ on $S$ is often called the \textit{mean measure} of the Poisson process $\Phi$.
When $S=\mathbb{R}^d$, the mean measure is in most interesting cases given in terms of \textit{intensity}.
This is a positive measurable function $\lambda$ on $S$, in terms of which $\mu$ is given by integrating $\lambda$ with respect to $d$-dimensional Lebesgue measure:
\begin{equation}
\label{eq:measMu}
\mu(A)=\int_A \lambda(x) \mathrm{d}x.
\end{equation}
If $\lambda$ is continuous at $x$, then \textit{eq.}\ref{eq:measMu} implies that for small neighbourhoods $A$ of $x$,
\begin{equation*}\notag
\mu(A)\approx\lambda(x)|A|,
\end{equation*}
where $|A|$ denotes the Lebesgue measure (length if $d=1$, area if $d=2$, volume if $d=3$) of A. Thus $\lambda(x)|A|$ is the approximate probability of a point of $\Phi$ falling in the small set $A$, and is larger in regions where $\lambda$ is large than in those where $\lambda$ is small.
In the special case when $\lambda$ is a constant, so that
\begin{equation}
\label{eq:muhom}
\mu(A)=\lambda|A|
\end{equation}
we speak of a \textit{uniform} or \textit{homogeneous} Poisson process.

In this paper, we assume that the sites of the Poisson-Voronoi diagrams are generated according to the particular case described by \textit{eq.}\ref{eq:muhom}.

As mentioned before, our aim is to find the distribution of the geometrical characteristics of the grains. In order to approximate these distributions, we generate a large sample of independent and identically distributed cells, more specifically \textit{typical cells}.
A typical Voronoi cell refers to a random polytope which loosely speaking has the same distribution as a randomly chosen cell from the diagram selected in such a way that every cell has the same chance of being sampled. Moreover, the distribution of the typical Poisson-Voronoi cell is by Slivnyak-Mecke formula \cite{moller1994lectures} the same as the Voronoi cell obtained when the origin is added to the point process $\Phi$. This corresponds formally to
\begin{equation*}\notag
\mathcal{C}=\{y\in\mathbb{R}^d:\|y\|\le\|y-x\|\,\, for \,\, all \, \, x\in\mathrm{\Phi}\}
\end{equation*}

Okabe et al. \cite{okabe2009spatial} synthesize previous research activity about the properties of Poisson Diagrams.
Despite the fact that distributions of several geometrical characteristics are already known, the distributions of the main features, especially in 3D, are not. We describe a simulation approach to approximate these distributions in the next section.

\section{Distribution of the geometrical properties of a typical cell}
\label{sec:distrTyp}
Given the complexity in finding explicit formulae for the distributions of the Poisson-Voronoi tessellation geometrical characteristics, especially in 3D, many authors used Monte Carlo methods to approximate these.
Among them Kiang \cite{kiang1966random}, Kumar and Kurtz \cite{kumar1995monte}, Lorz and Hahn \cite{lorz1993geometric}, M{\o}ller \cite{moller1994lectures}, Tanemura \cite{tanemura2003statistical} obtained numerical results for the moments of the distribution of volume, surface area, and number of faces of the grain in 3D.
They also give histogram estimates of these distributions and suggest approximations for them using various well known probability distributions.
For instance, for the volume distribution, before 1990 most authors used the Lognormal distribution for approximating the grain size distribution in polycrystals. Nowadays, more flexible distributions such as Gamma or Generalized Gamma are commonly used (e.g. \cite{tanemura2003statistical, kumar1995monte}).
Although these models fit the observed data rather well (as we will see in the next section) our approach allows to find an accurate approximation of the true distribution and which parametric distribution optimally fits the data.

The main idea is that, given a Poisson-Voronoi diagram generated by a Poisson point process $\Phi$ with intensity parameter $\lambda$, this $\lambda$ is the only parameter determining the distributions of the geometrical features of the grains.
Furthermore, the dependence of the distributions on the intensity parameter is via simple scaling of a `parent distribution', due to the following important scaling property of the Poisson process.

\begin{lemma}[Scaling Property]
\label{lem:scaling}
Let $\Phi=\{X_1,X_2,\dots\}$ be a Poisson process on $\mathbb{R}^d$ with intensity $\lambda=1$. Choose $\lambda>0$ and define $\Phi_\lambda=\{X_1/\lambda^{1/d}, X_2/\lambda^{1/d}, \dots\}$. Then $\Phi_\lambda$ is a Poisson process with intensity $\lambda$.
\end{lemma}

\begin{proof}
The fact that $\Phi_\lambda$ is a Poisson process is a special instance of the `Mapping theorem' \cite[see][Section 2.3]{kingman1993poisson}, using states space $S=T=\mathbb{R}^d$ and $f(x_1, x_2, \dots, x_d)=(x_1/\lambda^{1/d}, x_2/\lambda^{1/d}, \dots, x_n/\lambda^{1/d})$. Denoting the mean measure of $\Phi$ (Lebesgue measure) by $\mu_1$, the induced mean measure $\mu_\lambda$ of $\Phi_\lambda$ is given by
\begin{equation*}\notag
\mu_\lambda(B)=\mu_1(f^{-1}(B))=\int_{f^{-1}(B)}\mathrm{d}\mu_1(x)=\lambda\int_{B}\mathrm{d}\mu_1(x)=\lambda\mu_1(B)
\end{equation*}
\end{proof}

In the following subsections Lemma \ref{lem:scaling} will be used to see how the distributions of volume,  surface area and number of faces of the grains depend on the intensity parameter $\lambda$. 

\subsection{Grain Volume}
We first focus our attention on the grain volume distribution because of the direct relationship of this Poisson-Voronoi geometrical characteristic and the grain size distribution in microstructure characterization of materials.

Exploiting the properties of the Poisson process, the distribution for the normalized length of the Voronoi cell in 1D or size measure in 1D, can be shown to have density \cite{meijering1953interface}
\begin{equation*}\notag
f_{1D}(y)=4y \,\mathrm{exp}(-2y)\,\,\,\, 1_{[0,\infty)}(y)
\end{equation*}
In dimension $d>1$, it was conjectured that the area (2D) and the volume (3D) of the typical cell in a Poisson-Voronoi diagram may be distributed as the sum of two and three gamma variables with shape and scale parameters equal to 2 \cite{kiang1966random}, but \cite{Zaninetti20093223} and \cite{ferenc2007size} showed the conjecture to be false.
In 2D an analytic, though computationally challenging result is provided by Calka \cite{calka2003precise}, which gives an expression for the distribution of the area of the typical cell in 2D given the number of vertices.
In 3D, as we know so far, no analytic expression for the volume distribution exists.

\begin{lemma}
\label{lem:scaleVol}
 Denote by $F_\lambda$ the distribution function of the volume (length if $d=1$, area if $d=2$) of the typical cell of the Poisson-Voronoi diagram based on a homogeneous Poisson process on $\mathbb{R}^d$ with intensity parameter $\lambda>0$. Then, for all $x\ge 0$,
\begin{equation}
F_\lambda(x)=F_1(\lambda x)
\end{equation}
\end{lemma}

\begin{proof}
Let $\Phi$ be a homogeneous Poisson process on $\mathbb{R}^d$ with intensity 1.
Denote by $\mathcal{C}$ the typical cell of the Voronoi diagram based on this process. Fix $\lambda>0$ and consider the homogeneous Poisson process $\Phi_\lambda$ with intensity $\lambda$ as introduced in the statement of \textbf{Lemma \ref{lem:scaling}}. Then the typical cell in the Voronoi diagram based on $\Phi_\lambda$ is a scaled version of the typical cell of the Voronoi diagram based on $\Phi$, in the sense that it is given by $\mathcal{C}_\lambda=\{x/\lambda^{1/d}:x\in\mathcal{C}\}$. This means that the volume $V_\lambda$ of $\mathcal{C}_\lambda$ is exactly $\lambda^{-1}$ times the volume $V$ of $\mathcal{C}$.
Therefore, for $x\ge 0$,

\begin{equation*}\notag
F_\lambda(x)=\mathrm{P}(V_\lambda\le x)=\mathrm{P}\left(\frac{V}{\lambda}\le x\right)=\mathrm{P}(V\le \lambda x)=F_1(\lambda x)
\end{equation*}
\end{proof}

\subsection{Grain Surface area}
\begin{lemma}
\label{lem:scaleSA}
 Denote by $G_\lambda$ the distribution function of the surface area of the typical cell of the Poisson-Voronoi diagram based on a homogeneous Poisson process on $\mathbb{R}^3$  with intensity parameter $\lambda>0$. Then, for all $x\ge 0$,
\begin{equation*}\notag
G_\lambda(x)= G_1\left(\lambda^{\frac{2}{3}}x\right)
\end{equation*}
\end{lemma}
\begin{proof}
The argument follows the proof of \textbf{Lemma \ref{lem:scaleVol}}. Denote by $S_\lambda$ the surface area of $\mathcal{C}_\lambda$ and note that scaling of $\mathcal{C}_\lambda$ implies that $S_\lambda$ is $\lambda^{-\frac{2}{3}}$ times the surface area of $\mathcal{C}$, $S$.
Therefore
\begin{equation*}\notag
G_\lambda(x)=\mathrm{P}(S_\lambda\le x)=\mathrm{P}\left(\frac{S}{\lambda^{\frac{2}{3}}}\le x \right)=\mathrm{P}\left(S\le \lambda^{\frac{2}{3}}x\right)=G_1\left(\lambda^{\frac{2}{3}}x\right)
\end{equation*}
\end{proof}

\subsection{Number of grain faces}
Finally, another (discrete) property of interest regards the number of grain faces of the typical cell. It is clear that using $\Phi$ or $\Phi_\lambda$ (from \textbf{Lemma \ref{lem:scaling}}) as a basis for the Voronoi diagram, does not make any difference in the number of faces of the typical cell ($\mathcal{C}$ or $\mathcal{C}_\lambda$ respectively), leading to
\begin{lemma}
\label{lem:scaleFaces}
 Denote by $N_\lambda$ the distribution function of the number of faces of the typical cell of the Poisson-Voronoi diagram based on a homogeneous Poisson process on $\mathbb{R}^d$ with intensity parameter $\lambda>0$. Then, for all $x\ge 0$,
\begin{equation*}
N_\lambda(x)=N_1(x)\notag
\end{equation*}
\end{lemma}
The same lemma holds for \textit{number of corner points}, $n_v$.
In fact, exploiting the Euler-Poincar\'e  relation, it is possible to determine $n_v$, when the number of faces is known.
\section{Simulation Results}
\label{sec:simulres}
Now, we intend to approximate the real distribution function of the grain geometrical features, using the results obtained by a simulation based on 1\,000\,000 Voronoi diagrams.
Our simulation approach also allows to determine which of the usual parametric models provides the best approximation to the true distribution and its deviation from it.
We consider the volume, the surface area and the number of faces of a 3D Poisson Voronoi typical cell.
There exist two possible approaches well described in \cite{okabe2009spatial}:
\begin{enumerate}
\item generate a large number of points inside a bounded region $B$ according to $\Phi$, construct $\mathcal{V}_\Phi$ and measure the characteristics of all its cells.
\item generate a sequence of independent typical Poisson Voronoi cells, measure the characteristics of each and then aggregate them to obtain the required distributions.
\end{enumerate}
We follow the second approach. The reason of this choice derives from the convenience of having a sample of independent and identically distributed Voronoi cells such that we can know how close we are  to the real distribution. Moreover, we are able to control and eliminate the boundary effect that is present because the structure is actually only constructed on a bounded region.
For our objective it is important that only the distributions of the geometrical properties of the typical cell are needed using $\lambda=1$ in the simulations. By \textbf{Lemma \ref{lem:scaleVol}, \ref{lem:scaleSA}} and \textbf{\ref{lem:scaleFaces}}, the distributions based on diagrams with different intensities can be obtained by scaling.


We conduct our simulation approach using the software provided by TATA Steel.
 The algorithm is based on the half plane intersection, which is closely related to the original definition of a Voronoi tessellation. Each Voronoi cell is constructed separately by intersecting $n-1$ half spaces. A disadvantage is that this algorithm computes in $O(n^2\log n)$ time \cite{o1998computational}, while the most frequently used incremental algorithms can do it in $O(n^2)$ time. To speed up the computations, the algorithm has been extended with a filter, which determines which neighboring points of a generator point (nucleus) are needed for the Voronoi cell construction of this nucleus point. This filter is built in such a way that it first sorts $\sim 80\%$ of the points which are certainly needed for Voronoi cell construction. After that, the other $\sim 20\%$ of the points are checked to see if they give half plane intersection with the Voronoi cell under construction. With this filter the computational speed is improved to be better than $O(n\log n)$, which is the computational speed of the fastest algorithm by Fortune \cite{f1986sweepline}.\\
Then, we adopt the following Monte-Carlo procedure.\\
Repeat $1\,000\,000$ times:

\begin{description}
\item[Step 1]: Generate a 3D Poisson-Voronoi diagram with added generator point $(0,0,0)$ with $\lambda=1$;
\item[Step 2]: Determine the geometrical characteristics of the realizations of the typical Voronoi cell, the cell that contains the point $(0, 0, 0)$, $\mathcal{C}(0)$;
\end{description}
Then, aggregate the $1\,000\,000$ values. \\
The main graphical results are shown in Figures \ref{fig:volDensDistr}, \ref{fig:SADensDistr} and \ref{fig:FacesDensDistr}.

The values of the estimated densities of the previously mentioned geometrical characteristics are given  at \href{http://dutiosb.twi.tudelft.nl/~martina/}{http://dutiosb.twi.tudelft.nl/~martina.vittorietti/}. An R-package is under construction.
In Tables \ref{tab:moments} and \ref{tab:freq}, we report the estimated moments of the main geometrical characteristics and the estimated probabilities for the number of faces.
They are coherent with both the theoretical and numerical results obtained by other authors \cite{tanemura2003statistical,kumar1995monte}.

\clearpage
\begin{figure}[!h]
\centering
\subfloat[]{\includegraphics[width=7cm]{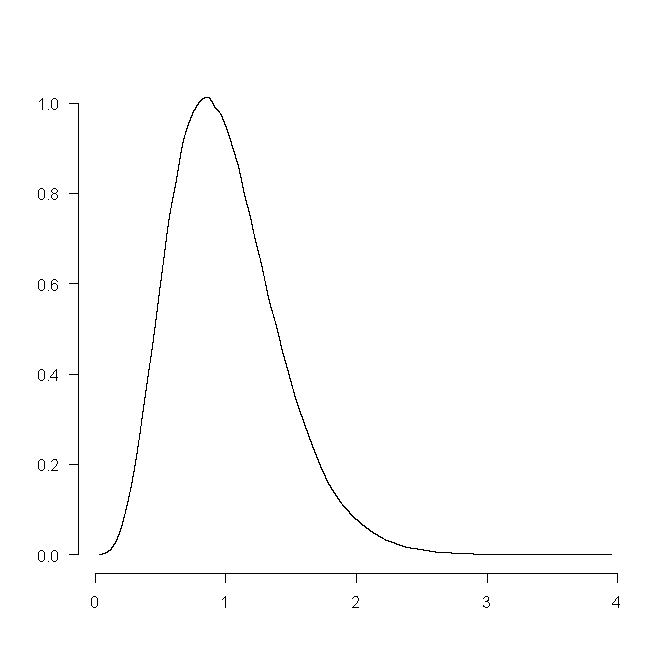}}
\subfloat[]{\includegraphics[width=7cm]{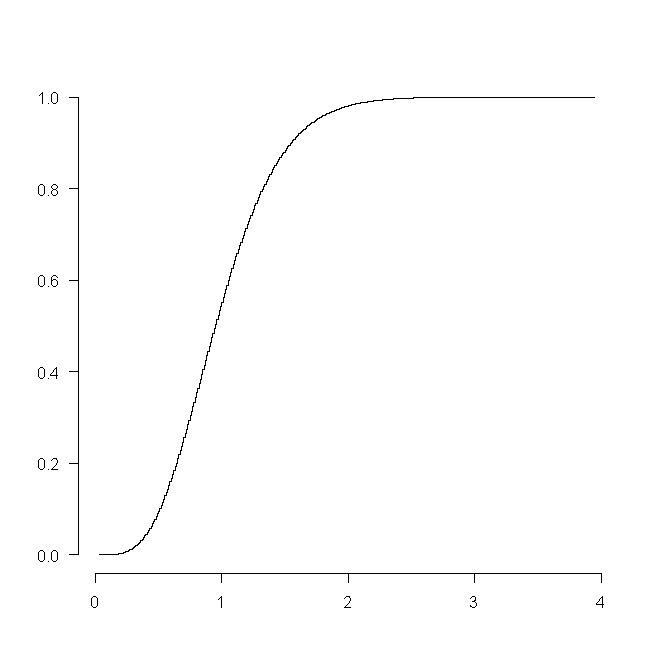}}
\caption{(a) Kernel density estimate (Epanechnikov kernel, cross validation bandwidth $h=0.05$) and (b) empirical cumulative distribution function of volume of 1\,000\,000 Poisson-Voronoi typical cells, $\lambda=1$}
\label{fig:volDensDistr}
\end{figure}

\begin{figure}[!h]
\centering
\subfloat[]{\includegraphics[width=7cm]{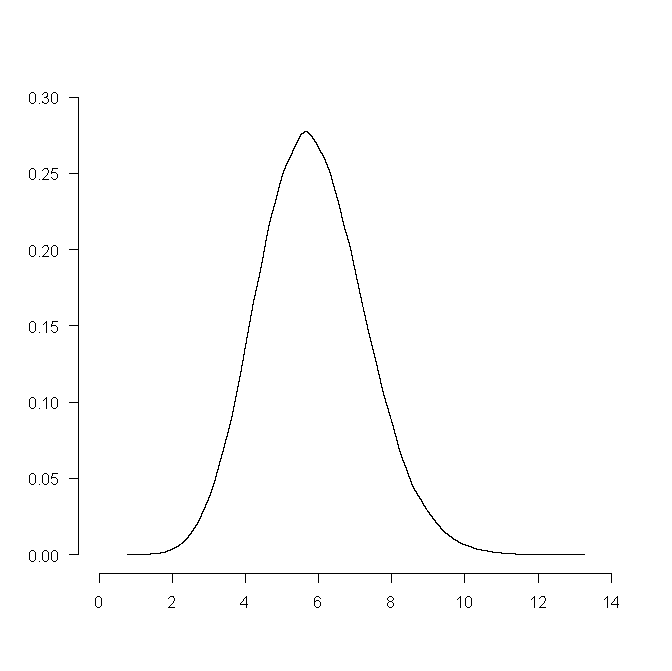}}
\subfloat[]{\includegraphics[width=7cm]{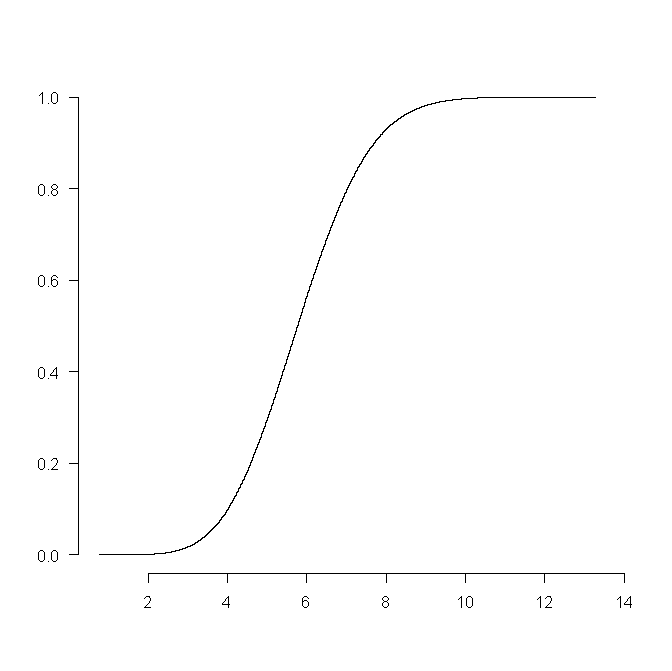}}
\caption{(a) Kernel density estimate (Epanechnikov kernel, cross validation bandwidth $h=0.25$) and (b) empirical cumulative distribution function of surface area of 1\,000\,000 Poisson-Voronoi typical cells, $\lambda=1$}
\label{fig:SADensDistr}
\end{figure}

\begin{figure}[!h]
\centering
\subfloat[]{\includegraphics[width=7cm]{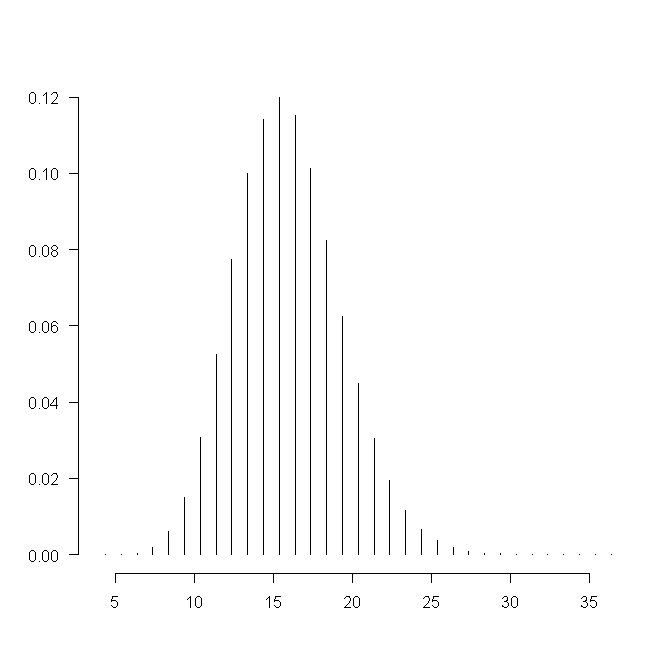}}
\subfloat[]{\includegraphics[width=7cm]{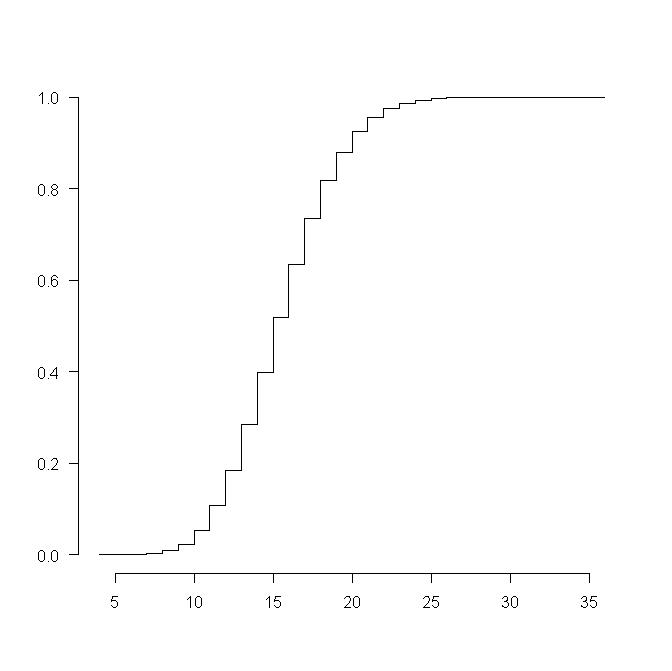}}
\caption{(a) Relative frequencies and (b) empirical cumulative distribution function of number of faces of 1\,000\,000 Poisson-Voronoi typical cells, $\lambda=1$}
\label{fig:FacesDensDistr}
\end{figure}

\begin{table}[!h]
\centering
\caption{Estimated moments of the geometrical features of 1\,000\,000 Poisson-Voronoi typical cells, $\lambda$=1}
\subfloat[][\hbox{Volume}]
{\begin{tabular}{rr}
  \hline
 $\mu_1$ & 1.00008  \\
  $\sigma$ &  0.41189 \\
  $\mu_2$ & 1.16981  \\
  $\mu_3$ & 1.55900 \\
  $\mu_4$ &  2.32340 \\
   \hline
\end{tabular}}
\subfloat[][\hbox{Surface area}]
{\begin{tabular}{rr}
  \hline
$\mu_1$ & 5.82670 \\
  $\sigma$ & 1.43821 \\
  $\mu_2$ & 36.01888 \\
  $\mu_3$ & 234.69091 \\
  $\mu_4$ & 1603.48468 \\
   \hline
\end{tabular}}
\subfloat[][\hbox{Number of faces}]
{\begin{tabular}{rr}
  \hline
$\mu_1$ & 15.53071 \\
 $\sigma$ & 3.33896 \\
  $\mu_2 $& 252.35173 \\
  $\mu_3 $& 4277.80397 \\
  $\mu_4 $& 75464.60519 \\
   \hline
\end{tabular}}
\label{tab:moments}
\end{table}
\newpage

\begin{table}[!h]
\centering
\caption{Distribution of the number of faces F of 1\,000\,000 Poisson-Voronoi typical cell,  $\lambda=1$}
\begin{tabular}{rrr}
  \hline
 F & $n_f$ & $p_f$ \\
  \hline
4 & 5 & 0.000005 \\
  5 & 35 & 0.000035 \\
  6 & 316 & 0.000316 \\
  7 & 1822 & 0.001822 \\
  8 & 6190 & 0.006190 \\
  9 & 15051 & 0.015051 \\
  10 & 30685 & 0.030685 \\
  11 & 52528 & 0.052528 \\
  12 & 77421 & 0.077421 \\
  13 & 100094 & 0.100094 \\
  14 & 114163 & 0.114163 \\
  15 & 120015 & 0.120015 \\
  16 & 115188 & 0.115188 \\
  17 & 101151 & 0.101151 \\
  18 & 82277 & 0.082277 \\
  19 & 62408 & 0.062408 \\
  20 & 44944 & 0.044944 \\
  21 & 30477 & 0.030477 \\
  22 & 19466 & 0.019466 \\
  23 & 11682 & 0.011682 \\
  24 & 6756 & 0.006756 \\
  25 & 3631 & 0.003631 \\
  26 & 1890 & 0.001890 \\
  27 & 975 & 0.000975 \\
  28 & 435 & 0.000435 \\
  29 & 224 & 0.000224 \\
  30 & 95 & 0.000095 \\
  31 & 52 & 0.000052 \\
  32 & 18 & 0.000018 \\
  33 & 3 & 0.000003 \\
  34 & 1 & 0.000001 \\
  35 & 1 & 0.000001 \\
  36 & 1 & 0.000001 \\
   \hline
\end{tabular}
\label{tab:freq}
\end{table}

\clearpage
\section{Parametric approximations to the distributions}
\label{sec:parappr}
Various proposals on how to estimate the distributions of the geometrical properties of the typical cell were put forward by several authors such as the Lognormal distribution \cite{smith1964mathematical,soderlund1998lognormal}, Generalized Gamma distribution with 2 \cite{vaz1988grain} or 3 parameters  \cite{tanemura2003statistical} and Rayleigh distribution \cite{pande1987stochastic}. Ferenc and N{\'e}da \cite{ferenc2007size} propose their own function for the volume distribution. Anyway, nobody was able to find an analytic expression for the grain geometrical characteristics distributions.

As noted in \cite{vaz1988grain} the use of the Lognormal distribution function for approximating grain size distribution lacks a solid physical basis and is not in general accurate.
Nowadays, the debate regards mostly the Generalized Gamma Distribution with 2 or 3 parameters, but until now no physical explanation for using one preferential distribution exists.
However, in view of the scaling properties described in the previous sections, it is natural to think that the distributions of the geometrical characteristics of the grain belong to a scale parametric family of distributions. Only then the distributions of the quantities {\it for all $\lambda$} can belong to the class.
Let us look, for instance, at the Lognormal distribution. Its probability density is given by
\begin{equation*}\notag
f(x|\mu,\sigma)=\frac{1}{x\sigma\sqrt{2\pi}}e^{-\frac{(\log(x)-\mu)^2}{2\sigma^2}}.
\end{equation*}
Let $\hat{\mu}_1$ and $\hat{\sigma}_1$ be the maximum likelihood estimates when $\lambda=1$ (based on the 1,000,000 simulated values). Define $f_\lambda(x)$ as:
\begin{equation*}\notag
f_\lambda(x)=\lambda f(\lambda x|\hat{\mu}_1,\hat{\sigma}_1)=\frac{1}{x\hat{\sigma}_1\sqrt{2\pi}}e^{-\frac{(\log(x)+log(\lambda)-\hat{\mu}_1)^2}{2\hat{\sigma}_1^2}}
\end{equation*}
which corresponds to a Lognormal distribution with parameter vector $(\hat{\mu}_1-log(\lambda),\hat{\sigma}_1^2)$.
Therefore, we have a log-addition scaling on the first parameter, which is not consistent with the $\lambda$-scaling that is found for real distributions.
Now let us consider the Generalized Gamma distribution. Its density function, parameterized according to \cite{stacy1965parameter}, is given by
\begin{equation}
\label{eq:paramGG}
f(x|a,b,k)=\frac{b x^{bk-1}}{\Gamma(k)a^{bk}}e^{-\left(\frac{x}{a}\right)^b}
\end{equation}
where $a$  and $b$ are the shape and the scale parameters.
Let $\hat{a}_1$, $\hat{b}_1$ and $\hat{k}_1$ be the maximum likelihood estimates for the parameters (based on the 1,000,000 simulated values) when $\lambda=1$.
Define $f_\lambda(x)$ as equal to:
\begin{equation*}\notag
f_\lambda(x)=\lambda f(\lambda x|\hat{a}_1,\hat{b}_1,\hat{k}_1)=\frac{\hat{b}_1x^{\hat{b}_1\hat{k}_1-1}}{\Gamma(\hat{k}_1)}\left(\frac{\lambda}{\hat{a}_1}\right)^{\hat{b}_1\hat{k}_1}e^{-\left(\frac{\lambda}{\hat{a}_1}x\right)^{\hat{b}_1}}
\end{equation*}
which corresponds to a Generalized Gamma with parameters $(\frac{\hat{a}_1}{\lambda},\hat{b}_1,\hat{k}_1)$. This suggests to look for a distribution that belongs to this scaling family.
 Special cases of this family are Gamma distribution with parameters $a$, $k$ and $b=1$; Weibull distribution with parameters $a$, $b$ and $k=1$.
Beside the parametrization in  \textit{eq}.\ref{eq:paramGG}, another one is provided by Prentice \cite{prentice1974log}. This is in general more stable in the estimation of the parameters but both parametrizations lead to the same estimates.
In the next subsections, we report the estimated parameters of the best Generalized Gamma approximations and we statistically compare the fits based on the Gamma distribution,  the Generalized Gamma distribution and the Lognormal distribution.
This comparison is based on two criteria:
\begin{itemize}
\item Supremum distance between two distribution functions:
\begin{equation}
D(F,G)=\sup_{x\in\mathbb{R}}|F(x)-G(x)|
\end{equation}
This distance is computed between the empirical distribution function $F_n$ of the sample and the maximum likelihood fit within the respective parametric families. The data are the 1,000,000 simulated values from the distribution of interest, with $\lambda=1$.
\item Total Variation distance between distributions with distribution functions $F$ and $G$ and densities $f$ and $g$ respectively on $\mathbb{R}$:
\begin{equation}
TV(F,G):=\sup_{A\in\mathcal{B}}|\int_A dF(x)-\int_A dG(x)|=\frac12\int |f(x)-g(x)|\,dx
\end{equation}
This distance is computed between the kernel estimate of the densities and maximum likelihood parametric fits based on the simulation results with $\lambda=1$.
\end{itemize}
Note that information on these distances are based on the data obtained with $\lambda=1$. If the estimates for more general values of $\lambda$ are obtained via the rescaling, these distances do, however, not change under this scaling, so the distances also hold for the other values of $\lambda$.

\begin{figure}[!h]
\centering
\includegraphics[width=8cm]{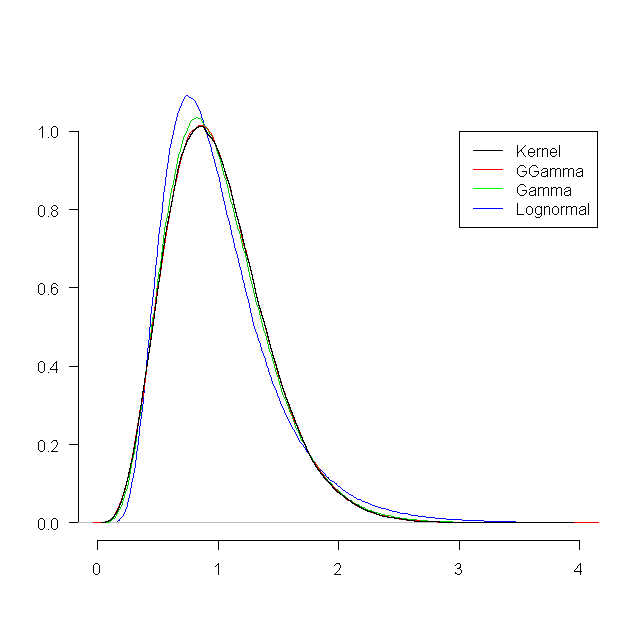}
\caption{Comparison of parametric approximations to the volume distribution of 1\,000\,000 Poisson-Voronoi typical cells, $\lambda=1$}
\end{figure}

\begin{figure}[!h]
\centering
\subfloat[]{\includegraphics[width=7cm]{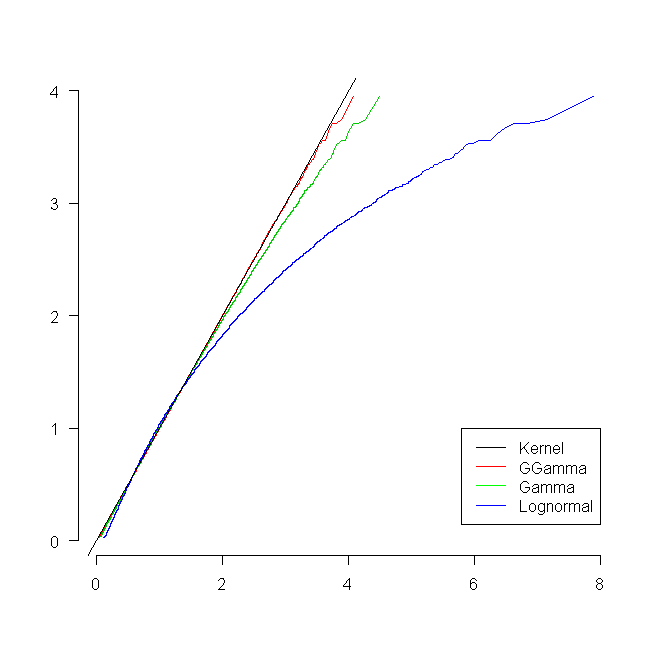}}
\subfloat[]{\includegraphics[width=7cm]{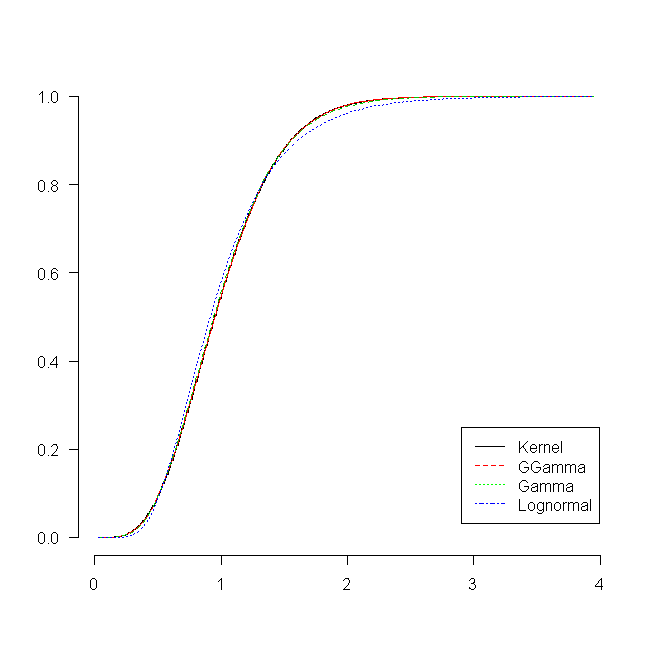}}
\caption{(a) QQplot and (b) cumulative distribution function comparison of parametric approximations to the volume distribution of 1\,000\,000 Poisson-Voronoi typical cells, $\lambda=1$}
\end{figure}

\begin{table}[!h]
\centering
\caption{Estimated Generalized Gamma parameters for volume distribution approximation, $\lambda$=1}
\begin{tabular}{lrrr}

&            $\hat{a}$ &$\hat{b}$&$\hat{k}$\\
\hline
Estimate&             0.380& 1.287  &3.583\\
Std. Error       &    0.005&       0.006&0.0322\\
 \hline
\end{tabular}
\end{table}

\begin{table}[!h]
\centering
\caption{Comparison of Gamma-, Generalized Gamma- and Lognormal approximations for volume distribution in terms of Supremum- and Total Variation distance}
\begin{tabular}{lrrr}

 &            Gamma &Generalized Gamma&Lognormal\\
  \hline
Supremum distance&      0.013  &         0.005&0.041\\
TV distance   &      0.018   &       0.005&0.089\\
 \hline
\end{tabular}
\end{table}

\begin{figure}[!h]
\centering
\includegraphics[width=8cm]{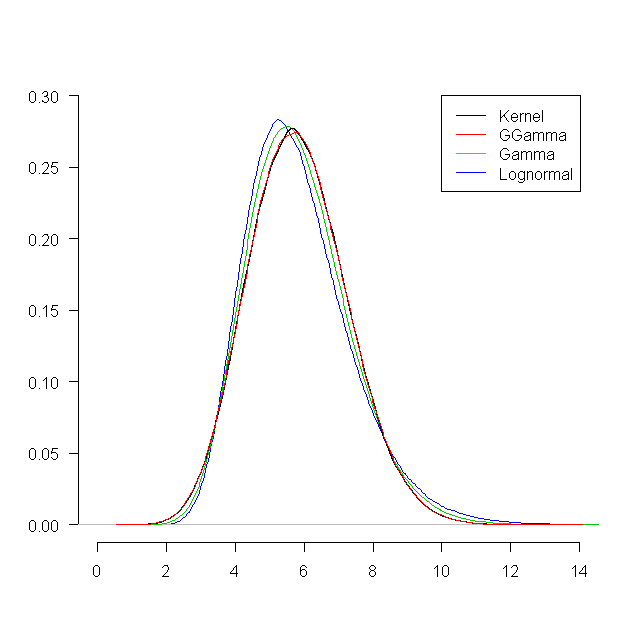}
\caption{Comparison of parametric approximations to the surface area distribution of 1\,000\,000 Poisson-Voronoi typical cells, $\lambda=1$}
\end{figure}

\begin{figure}[!h]
\centering
\subfloat[]{\includegraphics[width=7cm]{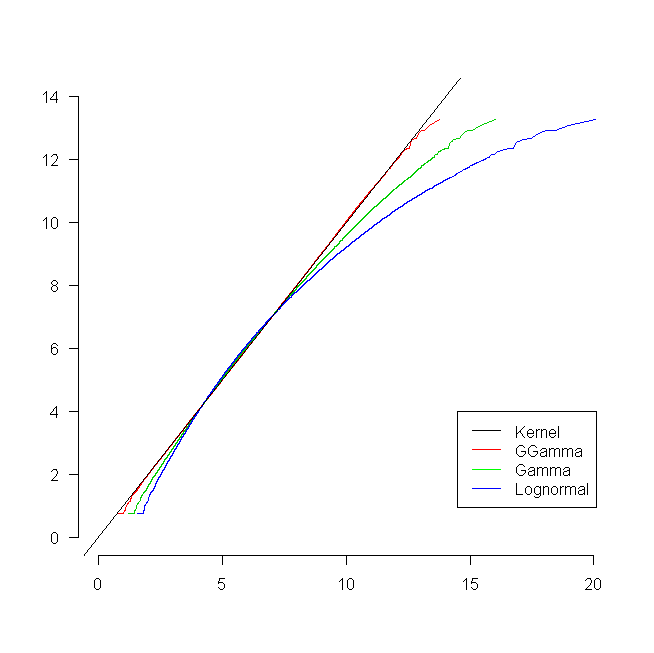}}
\subfloat[]{\includegraphics[width=7cm]{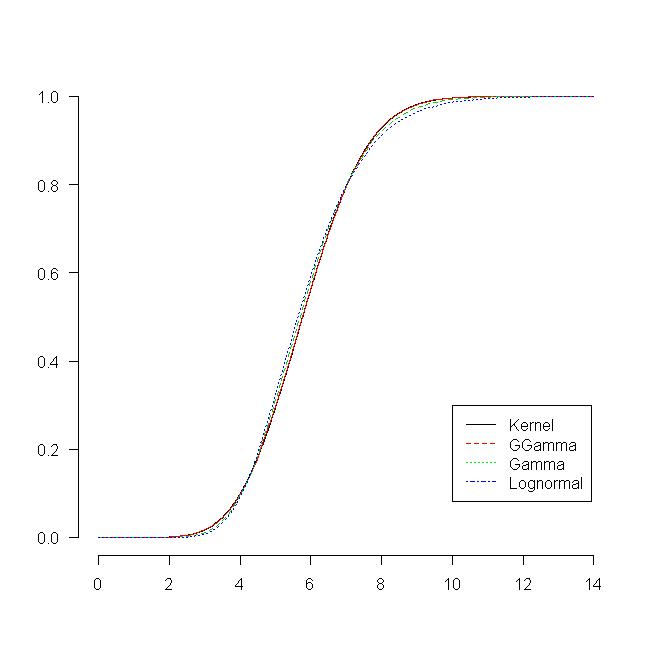}}
\caption{(a) QQplot and (b) cumulative distribution function comparison of parametric approximations to the surface area of 1\,000\,000 Poisson-Voronoi typical cells, $\lambda=1$}
\end{figure}

\begin{table}[!h]
\centering
\caption{Estimated Generalized Gamma parameters for surface area distribution approximation}
\begin{tabular}{lrrr}

&            $\hat{a}$ &$\hat{b}$&$\hat{k}$\\
\hline
Estimate&             3.174&2.102  &3.839\\
Std. Error   &          0.025& 0.011    &0.036\\
 \hline
\end{tabular}
\end{table}
\clearpage
\begin{table}[t]
\centering
\caption{Comparison of Gamma-, Generalized Gamma- and Lognormal approximations for surface area distribution in terms of Supremum- and Total Variation distance}
\begin{tabular}{lrrr}
  \hline
               &     Gamma  & Generalized Gamma&Lognormal\\
Supremum distance  &     0.020   &          0.002&0.037\\
TV distance      &    0.035    &          0.003&0.082\\
\hline
\end{tabular}
\end{table}

\section{Results and Discussion}
This work provides a very accurate approximation for the distributions of all the main geometrical characteristics of Poisson-Voronoi typical cell.
It is now possible to exploit our approximated distribution for generating observations of approximately every geometrical characteristics that define the grain size and for every possible value of $\lambda$.
Moreover, we prove that the Generalized Gamma distribution, with parameter $\hat{a}=0.380$, $\hat{b}=1.287$ $\hat{k}=3.583$ for the volume and $\hat{b}=3.174$, $\hat{a}=2.102$, $\hat{k}=3.839$  for the surface area is the best approximation in the class of the commonly used parametric distributions for grain size distributions.
Nevertheless, it is not the true underlying distribution.
In fact, the interpretation of the total variation distance as measure of quality allow us to say that using Generalized Gamma approximation for the grain volume distribution we could commit an error of about $0.5\%$ and about $0.3\%$ for the grain surface area. \\

As noted in the introduction, Poisson-Voronoi diagrams are interesting and widely applied, but for modeling microstructures only constitute the most basic case. Their use in microstructure characterization is not fully evaluated yet. 
Therefore, we want to test their applicability and then extend our work to more general and less understood Voronoi structures, such as Multi-level Voronoi diagrams \cite{p2015kok} or Laguerre-Voronoi tessellations \cite{l2008rlaguerre} in which the convenient scaling properties present in the Poisson-Voronoi diagrams are less obvious.

\section*{Acknowledgements}
This research was carried out under project number S41.5.14547b in the framework of the Partnership Program of the Materials innovation institute M2i (www.m2i.nl) and (partly) financed by the Netherlands Organisation for Scientific Research (NWO).

\bigskip
\noindent
\end{document}